\documentclass[conference,a4paper]{IEEEtran}

\usepackage{pgfplots}
\pgfplotsset{compat=newest}
\usepackage{tikz}
\usetikzlibrary{arrows,matrix,positioning}
\usepackage[utf8]{inputenc}
\usepackage[innermargin=0.545in,outermargin=0.65in,top=0.545in,bottom=.7in]{geometry}
\renewcommand{\baselinestretch}{0.95158} 
\usepackage[english]{babel}
\usepackage{epsfig}
\usepackage{amsmath, amssymb}
\usepackage{mathdots}
\usepackage{xspace}
\usepackage[noend]{algpseudocode}
\usepackage[linesnumbered,ruled,vlined,titlenumbered]{algorithm2e}

\usepackage{color}
\usepackage{cite}
\usepackage{booktabs}
\usepackage{verbatim}
\usepackage{enumitem}
\usepackage{colortbl}

\usepackage{amsthm}
\newtheorem{definition}{Definition}
\newtheorem{theorem}{Theorem}
\newtheorem*{remark}{Remark}

\newtheorem{problem}{Problem}

\usepackage{array}
\usepackage{multirow}
\usepackage{arydshln}

\def\ve#1{{\mathchoice{\mbox{\boldmath$\displaystyle #1$}}%
              {\mbox{\boldmath$\textstyle #1$}}%
              {\mbox{\boldmath$\scriptstyle #1$}}%
              {\mbox{\boldmath$\scriptscriptstyle #1$}}}}

\newcommand{\F}{\ensuremath{\mathbb{F}}}
\newcommand{\Fq}{\ensuremath{\mathbb{F}_q}}
\newcommand{\Fqm}{\ensuremath{\mathbb{F}_{q^m}}}

\newcommand{\tpub}{\ensuremath{t_{\mathsf{pub}}}}

\newcommand{\code}{\ensuremath{\mathcal{C}}}

\renewcommand{\a}{\ve{a}}

\renewcommand{\c}{\ve{c}}

\newcommand{\M}{\ve{M}}

\newcommand{\floor}[1]{\ensuremath{\left\lfloor{#1}\right\rfloor}}

\newcommand{\set}[1]{\ensuremath{\mathcal{#1}}}

\newcommand{\GF}[2]{\ensuremath{\mathbb{#1}_{#2}}}

\newcommand{\Gpub}{\ensuremath{\ve{G}_\mathsf{pub}}}
\newcommand{\CE}{\ensuremath{\mathcal{C}_\text{E}}}
\newcommand{\wt}{\textrm{wt}}
\newcommand{\supp}{\text{supp}}

\newcommand{\tikznode}[2]{%
	\ifmmode%
	\tikz[remember picture,baseline=(#1.base),inner sep=0pt] \node (#1) {$#2$};
	\else
	\tikz[remember picture,baseline=(#1.base),inner sep=0pt] \node (#1) {#2};%
	\fi}
\tikzset{%
	mybox_block/.style={rectangle,rounded corners,draw=black, thick,text width=1em,minimum height=2em,minimum width=4.75em,text centered},                  
	[highlight/.style={rectangle,rounded corners,fill=#1!15,draw,fill opacity=0.5,thick,inner sep=0pt},
	highlight/.default=gray],
	plot1/.style = {thick,
		dotted,
		mark=+}
}

\newif\ifcomment
\commenttrue

\newcommand{\dE}{d_\mathrm{E}}

\newcommand{\E}{\ve{E}}
\newcommand{\Y}{\ve{Y}}
\newcommand{\Code}{\mathcal{C}}

\newcommand{\Eset}{\mathcal {E}}

\newcommand{\basepoly}{b}

\newcommand{\C}{\ve{C}}
\newcommand{\R}{\ve{R}}
\renewcommand{\CE}{\Code_\mathrm{E}}

\renewcommand{\a}{\ve{a}}
\renewcommand{\b}{\ve{b}}

\newcommand{\IGC}{\ensuremath{\mathrm{I}\Gamma}}

\IEEEoverridecommandlockouts

\begin{document}

	\title{On Decoding and Applications of \\Interleaved Goppa Codes}
	\author{\IEEEauthorblockN{Lukas Holzbaur, Hedongliang Liu, Sven Puchinger, Antonia Wachter-Zeh}
		\IEEEauthorblockA{
			Institute for Communications Engineering, Technical University of Munich (TUM), Germany\\
			Email: \{lukas.holzbaur, lia.liu, sven.puchinger, antonia.wachter-zeh\}@tum.de
			\thanks{This work was supported by the Technical University of Munich -- Institute for Advanced Study, funded by the German Excellence Initiative and European Union 7th Framework Programme under Grant Agreement No. 291763 and the German Research Foundation (Deutsche Forschungsgemeinschaft, DFG) under Grant No. WA3907/1-1.
			}
	}}
	
	\maketitle
	
	\begin{abstract}
          Goppa Codes are a well-known class of codes with, among others, applications in code-based cryptography. In this paper, we present a \emph{collaborative decoding} algorithm for \emph{interleaved Goppa codes} (IGC). Collaborative decoding increases the decoding radius beyond half of the designed minimum distance. We consider wild Goppa codes and show that we can collaboratively correct more errors for binary Goppa codes than the Patterson decoder. We propose a modified version of the McEliece cryptosystem using wild IGC based on a recently proposed system by Elleuch et al., analyze attacks on the system and present some parameters with the corresponding key sizes.
	\end{abstract}
	
	\begin{IEEEkeywords}
		Interleaved Goppa codes, decoding, public-key cryptosystem, code-based cryptography, McEliece system
	\end{IEEEkeywords}
	
	\section{Introduction}
	G{oppa} codes \cite{GVD70} are a subclass of algebraic error-correcting codes called \emph{alternant codes}~\cite[Chapter~12]{Mac78}, which are subfield subcodes of generalized Reed--Solomon (RS) codes~\cite{RS60}. Therefore, every Goppa code of length $n$ over $\mathbb{F}_q$ is a subfield subcode of a generalized RS code in $\mathbb{F}_{q^m}^n$ and can be decoded with any RS decoder. 
	Alternatively, Goppa codes can be decoded by code specific algorithms, e.g., by solving a key equation with the Euclidean algorithm~\cite{decGC75keyeq}. Patterson~\cite{Pat75} introduced an algorithm with an extra ``key equation degree reduction'' step, which increases the decoding radius of binary Goppa codes. Barreto et al.~\cite{Barreto13} introduced a probabilistic algorithm which generalizes Patterson's algorithm over any prime field $\mathbb{F}_p$ to increase the decoding radius of Goppa codes from $\floor{\frac{r}{2}}$ to $\lfloor\frac{2}{p}r\rfloor$. Moreover, several list decoding approaches~\cite{List2011Augot, List2011Bern, ListBiGC2013} were proposed in order to decode Goppa codes beyond half the designed minimum distance. 
	
	Interleaved RS codes can be decoded almost up to the Singleton bound by collaborative decoding~\cite{schmidt2009collaborative}. Since interleaved Goppa codes are subcodes of interleaved RS codes, they can be decoded by any collaborative RS decoder. As an alternative, we present in this work the first collaborative decoder specifically for interleaved Goppa codes.
	
	In the second part of this work we consider the application of Goppa codes in code-based cryptography. The threat of quantum computers to the security of currently used public-key cryptosystems sparked an increased interest in post-quantum secure cryptosystems. One promising approach are code-based cryptosystems, such as the McEliece cryptosystem~\cite{McE78}. Besides being post-quantum secure, it also provides faster encryption and decryption than conventional public-key systems because algebraic error correcting codes offer efficient encoding and decoding algorithms. The downside of the McEliece cryptosystem is that for a given security level the key size is significantly larger than for currently used cryptosystems (e.g., for 128 bits security level, the key size of the original McEliece system is several hundred KB and for RSA $<1$ KB). This security level of the system depends heavily on the chosen code and several classes of codes have been proposed to decrease the key size. However, only Goppa codes have remained secure for a long time.

	In this work we introduce a new decoder for interleaved Goppa codes, based on Patterson's key equation~\cite{Pat75}. Further, we propose a repair and improvement of the interleaved McEliece scheme of~\cite{IGC2018molka} to secure the system against Tillich's attack~\cite{Tillich2018Persional} and present parameters for different security levels.

	\section{Preliminaries}\label{sec:preliminaries}
	\subsection{Notations}
	Let $\GF{F}{q}$	be a finite field of size $q$. Denote by $\ve{a}\in \GF{F}{q}^n$ a vector of length $n$ over $\GF{F}{q}$ and by $\ve{A}\in\GF{F}{q}^{a \times b}$ a matrix with $a$ rows and $b$ columns over $\GF{F}{q}$. 
	We denote the Hamming weight of a vector $\ve{a}$ by $\wt(\ve{a})$ and the number of non-zero columns of $\ve{A}$ by $\wt(\ve{A})$. 
	A linear code $\set{C}$ of length $n$, dimension $k$ and minimum distance $d$ over $\GF{F}{q}$ is denoted by $[n,k,d]_q$ or $[n,k]_q$. 
	\subsection{Goppa Codes}
	
	A Goppa code~\cite{GVD70} (see also~\cite{Mac78}) is defined by a \emph{locator set} $\set{L}$ and a \emph{Goppa polynomial} $g(x)$.
	\begin{definition}[Goppa Code]\label{def:GoppaCode}
		Let $q$ be a prime power and $m, n, r$ be some integers such that $rm\leq n\leq q^m$. Let $\set{L}=\{\alpha_0,\hdots,\alpha_{n-1}\}$ be a set of $n$ distinct elements of $\mathbb{F}_{q^m}$ and 
		$g(x)\in\GF{F}{q^m}[x]$ be a polynomial of degree $r$ such that $g(\alpha_i)\neq 0, \forall \alpha_i\in \set{L}$. 
		The Goppa code $\Gamma(\set{L},g)$ is defined as
		\begin{equation*}
		\Gamma(\set{L},g)=\left\{\ \ve{c}\ \Big|\  \sum\limits_{i=0}^{n-1}\frac{c_i}{x-\alpha_i} \equiv 0 \mod g(x),\; \forall \ve{c}\in\mathbb{F}_q^n \right\}.
		\end{equation*}
	\end{definition}
	If $g(x)$ has no multiple irreducible factors then $\Gamma(\set{L},g)$ is called a \emph{square-free} or \emph{separable} Goppa code. In addition, if $g(x)$ is an irreducible polynomial then $\Gamma(\set{L},g)$ is called an \emph{irreducible} Goppa code.
	A Goppa code $\Gamma(\set{L},g)$ as in Definition~\ref{def:GoppaCode} is a linear code over $\mathbb{F}_q$ of length $n = |\set{L}|$, dimension $k\geq n-mr$ and minimum distance $d\geq r+1$. For irreducible binary Goppa codes the distance is $d\geq 2r+1$. It is well-known that Goppa codes are subfield subcodes of $[n,n-r]_{q^m}$ generalized RS codes.
	\subsection{Wild Goppa Codes}\label{sec:wildGoppa}
	\emph{Wild Goppa codes}~\cite{WildGC76} are a subclass of Goppa codes and have been suggested for the \emph{Wild McEliece}~\cite{WildBern2011}.
	
	\begin{definition}[Wild Goppa Codes]\label{def:wildGC}
		Let $\Gamma(\set{L},g)$ be as in Definition~\ref{def:GoppaCode}. If $\basepoly(x)$ is a monic \textbf{square-free} polynomial in $\GF{F}{q^m}[x]$, the Goppa codes $\Gamma(\set{L},\basepoly^q)$ and $\Gamma(\set{L},\basepoly^{q-1})$ are called \emph{wild Goppa codes}.
	\end{definition}
	It has been shown in~\cite{WildGC76} that the wild Goppa codes $\Gamma(\set{L},b^q)$ and $\Gamma(\set{L},b^{q-1})$ are the same code of length $|\set{L}|=n$, dimension $k\geq n-rm$ and distance $d \geq \frac{q}{q-1}r +1$, where~$r=\deg(b(x)^{q-1})$ .
	\begin{remark}[Binary square-free Goppa codes]
		The well-known binary square-free Goppa codes of minimum distance $d\geq 2r+1$ are a subclass of wild Goppa codes.
	\end{remark}
	\subsection{Interleaved Goppa Codes}
	\begin{definition}[Interleaved Goppa Codes]
		Let $\Gamma(\mathcal{L},g)$ be a Goppa code as in Definition~\ref{def:GoppaCode}.
		An $\ell$-\emph{interleaved Goppa code} (IGC) $\IGC(\mathcal{L},g,\ell)$ is defined as
		\begin{equation*}
		\IGC(\mathcal{L},g,\ell)= 
		\left\lbrace
		\C \!= \!
		\begin{pmatrix}
		\ve{c}^{(1)}\\
		\vdots\\
		\ve{c}^{(\ell)}\\
		\end{pmatrix}\!, \! \forall \ve{c}^{(i)}\!\in \Gamma(\mathcal{L},g),  i=1,\dots,\ell
		\right\rbrace\!.
		\end{equation*}
	\end{definition}
	The advantage of interleaved codes is that if the errors occur in the same positions in all rows, e.g., because of burst errors on the channel, collaborative decoding of the rows can increase the decoding radius beyond half the (designed) minimum distance.
	Let $\R = \C+\E$ be the received word, where $\C \in \IGC(\mathcal{L},g,\ell)$ is a codeword and $\E \in \Fq^{\ell \times n}$ is an error matrix.
	By $\Eset$, we denote the set of indices of the non-zero columns of the error matrix $\E = [e_{ij}]$. The number of (burst) errors is thus given by $t =\wt(\E) := |\Eset|$.

	\section{Decoding of Interleaved Wild Goppa Codes}\label{sec:decIGC}
	
	In the following, we only consider interleaved codes that arise from wild Goppa codes, i.e., $g(x) = \basepoly(x)^q$ for some square-free polynomial $\basepoly(x)$ in $\Fqm[x]$.
	Recall that $r := (q-1) \deg \basepoly(x) = \tfrac{q-1}{q} \deg g(x)$ (e.g., $r = \deg \basepoly(x) = \frac{\deg g(x)}{2}$ for $q=2$).
	
	\subsection{Interleaving Patterson's Key Equation}
	
	We present a decoder for interleaved wild Goppa codes based on Patterson's decoder~\cite{Pat75}.
	The decoder solves a system of key equations which contains the following polynomials.
	
	\begin{definition}\label{def:ELP_EEP_Syndrome}
		Let $\R = [r_{ij}]$, $\E = [e_{ij}]$, and $\Eset$ be defined as above. For $i=1,\dots,\ell$, we define the \emph{error locator} $\Lambda(x)$, the \emph{$i^\mathrm{th}$ error evaluator} $\Omega_i(x)$ and the \emph{$i^\mathrm{th}$ syndrome} $S_i(x)$ polynomials as follows,
		\begin{align*}
		\Lambda(x) &:= \prod_{j \in \Eset} (x-\alpha_j) \ , \\
		\Omega_i(x) &:= \sum_{j \in \Eset} e_{ij} \prod_{\mu \in \Eset \setminus \{j\}} (x-\alpha_\mu) \ , \\
		S_i(x) &:= \sum\limits_{j=0}^{n-1}\frac{r_{ij}}{x-\alpha_j} \equiv \sum\limits_{j \in \Eset}\frac{e_{ij}}{x-\alpha_j} \mod g(x) \ .
		\end{align*}
	\end{definition}
	
	The goal of the decoder is to find the unknown error locator and evaluator polynomials from the known syndrome polynomials such that they fulfill the following relation.
	
	\begin{theorem}[System of Key Equations]\label{thm:key_equations}
		\begin{align*}
		\Omega_i(x) &\equiv \Lambda(x) S_i(x) \mod g(x) \ , \\
		\deg \Omega_i(x) &< \deg \Lambda(x) = |\Eset|
		\end{align*}
		for all $i=1,\dots,\ell$.
	\end{theorem}
	
	\begin{proof}
		This congruence relations and inequalities follow directly from the definition.
	\end{proof}
	
	Theorem~\ref{thm:key_equations} assumes the specific structure of $\Omega_i(x)$ and $\Lambda(x)$ given in Definition~\ref{def:ELP_EEP_Syndrome}, which makes direct solving of the key equations a non-linear problem.
	Instead, we solve the following linearized, well-studied, version of the problem.
	
	\begin{problem}\label{prob:shift_register_problem}
		Given $g(x), S_1(x),\dots,S_\ell(x) \in \Fqm[x]$, find $\lambda(x), \omega_1(x), \dots, \omega_\ell(x) \in \Fqm[x]$, not all zero, such that
		\begin{align}
		\omega_i(x) &\equiv \lambda(x) S_i(x) \mod g(x) \ , \label{eq:shift_register_congruence} \\
		\deg \omega_i(x) &< \deg \lambda(x) \ , \label{eq:shift_register_degree} \\
		\deg \lambda(x) &\text{ minimal} \ . \label{eq:shift_register_minimality}
		\end{align}
	\end{problem}
	
	\begin{remark}
		Problem~\ref{prob:shift_register_problem} is well-studied in literature, see the overview and relation to several decoding problems in~\cite{nielsen2013generalised,nielsen2013list}.
		A solution of the problem can be found in
		\begin{align*}
		\mathcal{O}\big(\ell^3 r \log^2(r) \log(\log(r)) \big)
		\end{align*}
		over $\Fqm$, see~\cite{nielsen2013generalised} (note that $\deg g(x) = \tfrac{q}{q-1}r \in \mathcal{O}(r)$).
	\end{remark}
	
	For $\ell=1$, we can prove that the solution of Problem~\ref{prob:shift_register_problem} agrees with the actual error locator and error evaluator polynomial up to a scalar factor, for up to $\tfrac{q}{q-1} \cdot \tfrac{r}{2}$ errors.
	
	\begin{theorem}
		Let $\ell=1$ and $|\Eset| \leq \tfrac{q}{q-1} \cdot \tfrac{r}{2}$. Let $\lambda(x),\omega_1(x) \in \Fqm[x]$ be a solution of Problem~\ref{prob:shift_register_problem} with input $g(x),S_1(x)$. Then, the solution fulfills
		\begin{equation*}
		\lambda(x) = c \cdot \Lambda(x) \quad \text{and} \quad \omega_1(x) = c \cdot \Omega_1(x)
		\end{equation*}
		for some non-zero constant $c \in \Fqm$.
	\end{theorem}
	
	\begin{proof}
		The proof works similar to \cite[Proposition~6.1]{roth2006}.
		We have $g(\alpha_i) \neq 0$ for all $i=0,\dots,n-1$, so $(x-\alpha_i) \nmid g(x)$, and $\gcd(\Lambda(x),g(x)) = 1$. Hence, the inverse of $\Lambda(x)$ modulo $g(x)$ exists and we can rewrite the key equation into
		\begin{equation*}
		\Lambda^{-1}(x) \Omega_i(x) \equiv S_i(x) \mod g(x) \ .
		\end{equation*}
		By \eqref{eq:shift_register_congruence}, we obtain
		$\omega_i(x) \equiv \lambda(x) \Lambda^{-1}(x) \Omega_i(x) \mod g(x)$, so
		\begin{align}
		\omega_i(x) \Lambda(x)  &\equiv \lambda(x) \Omega_i(x) \mod g(x) \ . \label{eq:key_equation_reformulated}
		\end{align}
		By definition, the degrees of both sides of the congruence are
		\begin{equation*}
		< 2 |\Eset| \leq \tfrac{q}{q-1} \cdot r = \deg g(x) \ ,
		\end{equation*}
		so we can omit the modulo operation. Hence,
		\begin{equation}
		\omega_i(x) \Lambda(x) = \lambda(x) \Omega_i(x) \ . \label{eq:key_equation_reformulated_equation}
		\end{equation}
		Furthermore, for $\ell=1$, we have $\Omega_1(\alpha_i) \neq 0$ for all $i \in \Eset$. Hence, $(x-\alpha_i) \nmid \Omega_1(x)$ and $\gcd(\Omega_1(x), \Lambda(x)) = 1$.
		By \eqref{eq:key_equation_reformulated_equation}, we must have $\Lambda(x) \mid \lambda(x)$.
		Since $\Lambda(x)$ and $\Omega_1(x)$ satisfy conditions \eqref{eq:shift_register_congruence} and \eqref{eq:shift_register_degree}, and $\lambda(x)$ is of minimal degree satisfying the conditions, we must have $\deg \lambda(x) \leq \deg \Lambda(x)$. Hence,
		\begin{equation}
		\lambda(x) = c \cdot \Lambda(x)
		\end{equation}
		for some non-zero scalar $c \in \Fqm$. We obtain $\omega_1(x) = c \cdot \Omega_1(x)$ from \eqref{eq:key_equation_reformulated_equation}.
	\end{proof}
	
	\begin{remark}
		For $\ell=1$, the system of key equations in Theorem~\ref{thm:key_equations} is equivalent to Patterson's key equation \cite[Equation~(3)]{Pat75} with $g(x) = b(x)^{q-1}$ instead of $g(x) = b(x)^{q}$. Since we use wild Goppa codes here, where $\Gamma(\set{L},b(x)^{q-1}) = \Gamma(\set{L},b(x)^{q})$, we can circumvent the ``reduction step'' in Patterson's decoder \cite[Algorithm~4]{Pat75} and directly decode up to $\tfrac{q}{q-1} \cdot \tfrac{r}{2}$ errors uniquely. This enables us to ``interleave'' our key equation which is not possible with the ``reduced key equation'' in \cite[Section~V]{Pat75}.
		Furthermore, for $q>3$, we can decode more errors than the algorithm in~\cite{Barreto13} uniquely. Note that both decoders are probabilistic and similar to~\cite{Barreto13} we have to rely on simulation results to determine the decoding failure probability (see Section~\ref{sec:SimulationResults}).
	\end{remark}
	
	By counting the number of unknowns (coefficients of $\lambda(x)$ and $\omega_i(x)$) and equations of the linear system given by the coefficients of the left- and right-hand side of the congruence relation, one can see that Problem~\ref{prob:shift_register_problem} can only have a unique minimal solution with
	\begin{align}
	\lambda(x) = c \cdot \Lambda(x) \text{ and }
	\omega_i(x) = c \cdot \Omega_i(x) \label{eq:unique_solution}
	\end{align}
	for some non-zero scalar $c \in \Fqm$ for all $i$ if
	\begin{align}\label{eq:IntDecodingRadius}
	|\Eset| \leq\frac{\ell}{\ell+1} \cdot \frac{q}{q-1} \cdot r =: t_{\max} \ .
	\end{align}
	Our simulation results indicate that below this maximal decoding radius, most of the error matrices $\E$ of weight at most $\wt(\E) = t_{\max}$ can be decoded by our algorithm (i.e., any solution of Problem~\ref{prob:shift_register_problem} fulfills \eqref{eq:unique_solution}).
	More precisely, the results indicate that the number of error patterns for which decoding fails or miscorrects decreases exponentially in the value $t_{\max}-t$, where $t$ is the actual number of errors.
	
	\begin{remark}
		As an alternative to the decoder presented above, we can directly decode in an interleaved variant of the GRS supercode of the used Goppa code (with minimum distance $\deg g(x)$).
		We can use all known decoding algorithms for these interleaved codes, e.g.,~\cite{krachkovsky1997decoding,bleichenbacher2003decoding,schmidt2009collaborative} (or the more advanced algorithms in~\cite{wachterzeh2014decoding,coppersmith2003reconstructing,parvaresh2007algebraic,cohn2013approximate,puchinger2017decoding}, which we will not consider in this paper).
		The algorithms in~\cite{krachkovsky1997decoding,bleichenbacher2003decoding,schmidt2009collaborative} yield the same maximal decoding radius as the interleaved Patterson decoder described above.
	\end{remark}
	
	\subsection{Simulation Results}\label{sec:SimulationResults}
	
	Since the interleaved decoding radius exceeds the unique decoding radius, decoding fails with a certain probability. For interleaved RS codes an upper bound on the probability of decoding failure was derived in~\cite{schmidt2009collaborative}. However, even though interleaved Goppa codes are subfield subcodes of interleaved generalized RS codes, this bound does not hold for the former, as it assumes random error patterns from $\GF{F}{q^m}$, while the error patterns in the case of interleaved Goppa codes are only from the subfield $\GF{F}{q}$. A bound for interleaved subfield subcodes is an open problem left for future work and we rely on simulation results to support our conjecture that decoding will succeed with high probability.
	
	Figure~\ref{fig:simulationResults} shows the simulation results for an $[127,85,\geq 13]_2$ wild IGC for $\ell=2$ and $\ell =5$. As it is well known that the rank of the error matrix is related to the failure probability (see, e.g.,~\cite{metzner1990general}), the probability of decoding failure for full-rank error matrices is also shown. The results confirm that the bound of~\cite{schmidt2009collaborative} does not hold for subfield subcodes, as it is clearly exceeded by the probability of decoding failure of the IGC, regardless of the rank of the error matrix. However, the seemingly exponential decay in probability of decoding failure supports our conjecture.
	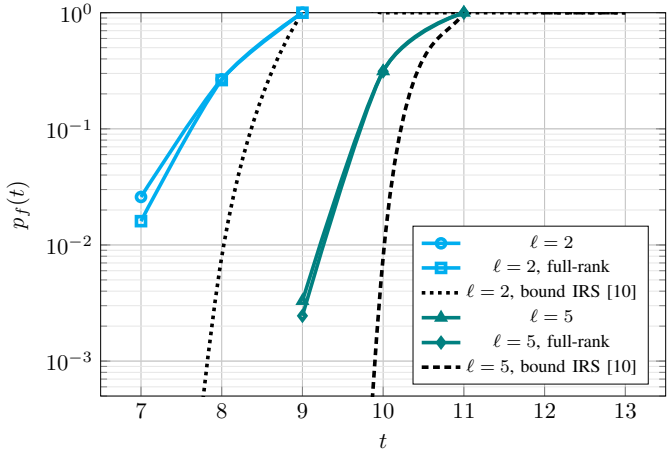
\begin{figure}
		\centering
		\resizebox{\linewidth}{!}{\begin{tikzpicture}
\begin{semilogyaxis}[
   width=\linewidth,
   height=6cm,
   scale only axis,
   grid=both,
   grid style={line width=.1pt, draw=gray!20},
   major grid style={line width=.2pt,draw=gray!50},
   xmin=6.5, xmax=13.5,
   xtick={4,5, ...,13},
   xmajorgrids,
   ymin=5e-4, ymax=1,
   xlabel={$t$},
   ylabel={$p_f(t)$},
   enlargelimits=false,
   legend pos = south east,
   legend style ={font = \footnotesize},
    	smooth
        ]

\addplot+[color=cyan, mark = o, ultra thick] coordinates{
  (3, 0.00000) (4, 0.00000) (5, 0.00000) (6, 0.00000) (7, 0.025907) (8, 0.26708) (9, 1.0000)
};
\addlegendentry{$\ell=2$};

\addplot+[color=cyan, mark = square, ultra thick] coordinates{
  (3, 0.00000) (4, 0.00000) (5, 0.00000) (6, 0.00000) (7, 0.016030) (8, 0.26284) (9, 1.0000)
};
\addlegendentry{$\ell=2$, full-rank};



\addplot+[color=black, mark = none, dotted, ultra thick] coordinates{
  (3, 1.9414e-34) (4, 4.0718e-28) (5, 8.5396e-22) (6, 1.7910e-15) (7, 3.7562e-9) (8, 0.0078778) (9, 1) (10, 1) (11, 1) (12, 1) (13, 1)
};
\addlegendentry{$\ell=2$, bound IRS \cite{schmidt2009collaborative}};

\addplot+[color=teal, mark = triangle, ultra thick] coordinates{
  (3, 0.00000) (4, 0.00000) (5, 0.00000) (6, 0.00000) (7, 0.00000) (8, 0.00000) (9, 0.0032972) (10, 0.31324) (11, 1.0000)
};
\addlegendentry{$\ell=5$};

\addplot+[color=teal, mark = diamond, ultra thick] coordinates{
  (3, 0.00000) (4, 0.00000) (5, 0.00000) (6, 0.00000) (7, 0.00000) (8, 0.00000) (9, 0.0024534) (10, 0.31271) (11, 1.0000)
};
\addlegendentry{$\ell=5$, full-rank};



\addplot+[color=black, mark = none, ultra thick] coordinates{
  (3, 2.4739e-91) (4, 1.0880e-78) (5, 4.7852e-66) (6, 2.1045e-53) (7, 9.2559e-41) (8, 4.0708e-28) (9, 1.7903e-15) (10, 0.0078740) (11, 1) (12, 1) (13, 1)
};
\addlegendentry{$\ell=5$, bound IRS \cite{schmidt2009collaborative}};

   \end{semilogyaxis}
\end{tikzpicture}%

}
		\caption{Probability of decoding failure of an $[127,85,\geq13]_2$ (unique decoding radius $=6$) wild IGC for $\ell=2$ and $\ell=5$ compared to the bound from~\cite{schmidt2009collaborative} on probability of decoding failure of the corresponding IRS supercode in $\GF{F}{2^7}$. For each point $>2000$ iterations were performed.} \label{fig:simulationResults}
		\vspace{-0.4cm}
	\end{figure}

	\subsection{$\ell$-Interleaved Subfield Subcodes in $\Fq$ vs. Codes in $\F_{q^\ell}$}
	
	Goppa codes are subfield subcodes of RS codes in $\Fqm$ and can be constructed over any field $\F_{q^\gamma}$ with $\gamma \mid m$. We first consider the case $q>2$. An $\ell$-interleaved Goppa code $\code_{\IGC}$ over $\Fq$ with $r_{\IGC}=\deg(g_{\IGC}(x))$ is of rate $R_{\IGC}\approx \frac{n-r_{\IGC}m}{n}$ and each codeword is of length $\ell n$ over $\Fq$. For $\ell | m$, a Goppa code $\code_\Gamma$ with Goppa polynomial $g_{\Gamma}(x)$ over $\F_{q^\ell}$ is of rate $R_\Gamma \approx \frac{n-r_{\Gamma}\frac{m}{\ell}}{n}$ and has codeword size $n \log_q q^\ell = \ell n$ over $\Fq$. To obtain the same rate for both codes, i.e., $R_{\Gamma} = R_{\IGC}$, the degree $r_\Gamma$ of $g_\Gamma(x)$ has to be chosen as
	\begin{align*}
	\frac{n-r_{\Gamma}\frac{m}{\ell}}{n} &= \frac{n-r_{\IGC}m}{n} \; \Rightarrow \;  r_\Gamma = \ell r_{\IGC} \ .
	\end{align*}
	For appropriately chosen Goppa polynomials the distances of the codes are $d_{\IGC} = r_{\IGC}+1$ and $d_\Gamma = r_\Gamma+1 = \ell r_{\IGC}+1$ respectively. Comparing the decoding radii for interleaved decoding of $\code_{\IGC}$ and bounded minimum distance decoding of $\code_{\Gamma}$ gives
	$t_{\IGC} = \frac{\ell}{\ell+1} (d_{\IGC}-1) > \frac{d_\Gamma-1}{2}  = t_{\Gamma}$, which implies $\ell <1$.
	It follows that, in general, the decoding radius of an $\ell$-interleaved Goppa code is not larger than the unique decoding radius of the corresponding code over a larger field with the same codeword size and code rate. The only exception are $\ell=2$ interleaved binary Goppa codes with square-free Goppa polynomial. These Goppa codes are of distance $d_{\IGC} = 2r_{\IGC} +1$ and hence the radius is increased for $\ell < 3$.
	
	It follows that if the sole motivation of interleaving is increasing the decoding radius regardless of the size of the generator and parity check matrix, it is generally advantageous to use a Goppa code over $\F_{q^\ell}$ instead of $\ell$-interleaving a Goppa code over $\F_q$, with the exception of $q=\ell=2$.
	However, as we will see in the next section, interleaved Goppa codes do have an application in code-based cryptography.

	\section{Application: Improvement and Reparation of a Cryptosystem based on Interleaved Goppa Codes}\label{sec:IntMcE}

	Recently,~\cite{IGC2018molka} proposed a variant of the McEliece cryptosystem based on interleaved Goppa codes.
	The idea is that the public key is an obfuscated generator matrix of a Goppa code and the ciphertext is a corrupted codeword of a corresponding interleaved code. 	Since the interleaved code can correct more (burst) errors than the original code, the level of security against generic decoding (e.g., information-set decoding), which usually determines the security level, is increased. On the other hand, structural attacks remain as hard as on the original Goppa-code-based system.
	Hence, smaller key sizes than in the original McEliece cryptosystem can be achieved.
	We modify the new system by using wild Goppa codes, which further increase the decoding radius. We also consider several attacks and propose a repair method and restrictions on parameters to avoid the attacks.
	
	\subsection{System Description}\label{sec:intsysdesc}
	
	Alice generates the key pair: \emph{public key} $(\Gpub,\tpub,\ell)$ and \emph{private key} $(\ve{S}, \ve{P}, \mathcal{D})$, where $\mathcal{D}$ is an efficient decoder for the $\ell$-interleaved wild Goppa code with generator matrix ${\ve{G}=\ve{S}^{-1} \Gpub \ve{P}^{-1} \in \F_q^{k\times n}}$ correcting up to $\tpub = t_{\max} = \frac{\ell}{\ell+1} \cdot \frac{q}{q-1} \cdot r$ errors.
	
	Bob encrypts the secret message $\M\in\GF{F}{q}^{\ell\times k}$ into a ciphertext $\Y\in\GF{F}{q}^{\ell \times n}$ by $\Y=\M\Gpub+\E$, where $\E\in\GF{F}{q}^{\ell\times n}$ is a full-rank random matrix with $\tpub$ non-zero columns. 
	Alice retrieves the secret messages by $\hat{\M}=\mathcal{D}(\Y\ve{P}^{-1})\ve{S}^{-1}$.
	\subsection{Decoding Attacks}\label{sec:attacks}
	\subsubsection{Finding the Low-Weight Codewords Attack}
	Consider the following three codes
	\begin{equation*}
	\Code := \left\langle \Gpub \right\rangle, \quad
	\Code' := \left\langle \begin{bmatrix} 
	\Gpub \\
	\Y
	\end{bmatrix}\right\rangle, \quad \text{ and } \quad \CE := \left\langle \E \right\rangle \ .
	\end{equation*}
	Obviously, $\Code' = \Code + \CE$ since we can perform row operations to get
	\begin{equation*}
	\begin{bmatrix}
	\Gpub \\
	\Y
	\end{bmatrix} \sim \begin{bmatrix}
	\Gpub \\
	\E
	\end{bmatrix} \ .
	\end{equation*}
	Hence, we have
	$d(\Code') \leq d(\CE) =: \dE$.
	Finding several words of weight $\dE$ in $\Code'$ might reveal error positions, thereby allowing information set decoding attacks (ISD) with less error positions. Assuming the worst case, i.e., all found words of weight $\dE$ belong to $\CE$ and the union of their support is the set of error positions, this gives an attack whose work factor is determined by algorithms for finding codewords of weight $\dE$ in a linear code. 
	
	Note that we can only guarantee $d(\Code') \leq \dE$.
	In principle, there might be codewords in $\Code'$ of smaller weight.
	Such codewords would always be of the form
	$\c = \a + \b$,
	where $\a \in \Code \setminus \{\ve{0}\}$, $\b \in \CE \setminus \{\ve{0}\}$, and $\wt(\a) < \dE + \tpub$. 
	Hence, the probability that such codewords exist depends on the weight distributions of the codes $\Code$ and $\CE$ (e.g., how many codewords $\a$ of weight $\wt(\a) < \dE + \tpub$ exist).
	Furthermore, even if such words $\c$ exist and are found by an attack, it needs to be studied whether $\c$ would reveal some of the error positions. Note that this is a general problem of any McEliece system correcting beyond the unique decoding radius, \emph{e.g.}, through list decoding~\cite{bernstein2008attacking,barbier2011key}.
	
	\subsubsection{Finding the Support of the Subcode Attack}
	Tillich pointed out in~\cite{Tillich2018Persional} that since the code $\CE[n,\ell]$ is a subcode of $\set{C}'[n,k+\ell]$ and $|\supp(\CE)|=\tpub$, where $\supp(\CE)=\{i\; ; \; \exists \ve{c}\in \CE, c_i\neq 0 \}$, one can reveal the error positions by finding $\supp(\CE)$.
	This problem has been studied by Otmani and Tillich for the binary case in~\cite{Tillich2011Attack}, which gives a very efficient attack if the parameter $p$ chosen to fulfill $2p\geq d_{GV}(\frac{\tpub-\ell}{n-\ell}(k+\ell+l),\ell)$ is small for some integer~$l$ (e.g., \cite{Tillich2011Attack} gives $1\leq p \leq 4$ as a typical range for $p$). For the parameters presented in Table~\ref{tab:keysize} this is not the case and the complexity of this attack is far from causing a security bottleneck. However, this attack needs to be considered when choosing the system parameters. 
	
	\subsection{Repair}
	Our repair is based on the idea of choosing the rows of the error matrix $\E$ as the basis of a code with large minimum distance $\dE$. The rows of $\E$  have to be chosen to be linearly independent to prevent brute-forcing linear combinations of the rows of $\Y$, resulting in error-free linear combinations of the codewords (i.e., rows of $\Y-\E$), which might reveal part of the message.
	
	Since $\E$ has only $\tpub$ non-zero columns, we choose the submatrix $\E'$ of $\E$, consisting of these columns, to be a generator matrix of a code with parameters
	\begin{equation*}
	\CE'[\tpub,\ell,\dE] := \langle \E' \rangle.
	\end{equation*}
	Thus, the overall error code $\CE$ has parameters $[n,\ell,\dE]_q$.
	
	\begin{remark}
		Note that since the code $\CE'$ is required to have specific properties, it might have to be considered public, e.g., if there are only few known constructions for the desired $\dE$. Then revealing the error positions in $\Y$ is equivalent to determining the permutational equivalence of $\CE$ and a subcode of $\set{C}'$, which has been shown to be an NP-complete problem~\cite{BGK17}. Nevertheless it needs to be studied if this could lead to a more efficient attack than finding low weight codewords. To avoid this kind of attack it is also possible to choose $\CE$ at random from some large family of codes. 
	\end{remark}
	\subsection{Measure of Security Level (SL)}
	
	To determine the SL of the original McEliece cryptosystem and our system (see Table~\ref{tab:keysize}), we use the currently fastest algorithm over arbitrary $\Fq$ presented in~\cite{interlando2018generalization}. Another recent algorithm~\cite{gueye2017generalization} might yield smaller security levels but that needs to be further verified.
	Both algorithms are generalizations of several important improvements of information-set decoding attacks since 2011: \cite{bernstein2011smaller} for \cite{interlando2018generalization} and \cite{hirose2016may,may2015computing,meurer2013coding,becker2012decoding,may2011decoding} for \cite{gueye2017generalization}. 
The SL of our repaired system is calculated with $\dE$ rather than $\tpub$, since any non-trivial linear combination of the received words (rows of the received matrix $\Y$) are codewords corrupted by errors of weight at least $\dE$.
	\subsection{Parameter Choice}\label{sec:params}
	In order to improve upon the original or Wild McEliece system, the work factor ISD~\cite{gueye2017generalization} must be larger than the one of generic decoding of the original system. Neglecting the difference in dimension ($k$ compared to $k+\ell$), this condition translates to
	\begin{equation}\label{eq:DistanceGain}
	\dE > \left\lfloor \frac{1}{2} \cdot \frac{q}{q-1}r \right\rfloor \ .
	\end{equation}
	In the following, we analyze for which parameters $q, \ell$ and $r$ such a linear code $\CE'$ exists.
	We start with the negative result that there is no improvement for $q=2$.
	
	\begin{theorem}\label{thm:noGainq2}
		For $q=2$, the work factor cannot be increased by interleaving.
	\end{theorem}
	\begin{proof}
		With (\ref{eq:IntDecodingRadius}) and (\ref{eq:DistanceGain}) we get
		\begin{equation*}
		\tpub < \frac{2\ell}{\ell+1} d_E .
		\end{equation*}
		as a necessary condition for an improvement compared to the original McEliece cryptosystem in terms of the code parameters of  $\CE'[\tpub,\ell,\dE]$.
		By the Griesmer bound~\cite{griesmer1960bound} the relation
		\begin{equation*}
		\tpub \geq \sum_{i=0}^{\ell-1} \left\lceil \frac{\dE}{2^i} \right\rceil \geq \sum_{i=0}^{\ell-1} \frac{\dE}{2^i} = \dE(2-2^{-(\ell-1)})
		\end{equation*} 
		holds and it follows that there can only be an improvement if
		\begin{align*}
		\tpub < \frac{2\ell}{\ell+1} \dE &\leq \frac{2\ell}{(\ell+1)(2-2^{-(\ell-1)})} \tpub \\
		\Leftrightarrow \quad 1 &< \frac{\ell}{(\ell+1)(1-2^{-\ell})} \\
		2^{\ell} &< \ell +1\ ,
		\end{align*}
		which is only the case for $\ell=1$.
	\end{proof}
	
	Larger fields, $q>2$, provide more flexibility in the code parameters.
	For $q\geq \tpub$, we can even achieve $\dE = \tpub-\ell+1$ (which is the maximal possible $\dE$ due to the Singleton bound) using an MDS code, but also for smaller field sizes there are codes with sufficiently large minimum distance $\dE$.
	For large $\tpub$, we could use asymptotically good sequences of codes, e.g., AG codes~\cite{shum2001low} over small fields.
	For small values of $\tpub$, we can use tables of good codes, e.g., \underline{CodeTables}~\cite{CodeTable}.
	For instance:
	\begin{itemize}
		\item For parameters $q=3$, $\ell=7$ and $\tpub = 110$ there is a $[110,7,\dE=70]_3$ code, while the unique decoding radius for these parameters is $\left\lfloor \frac{1}{2} \cdot \frac{q}{q-1}r \right\rfloor =63$.
		\item For parameters $q=4$, $\ell=9$ and $\tpub = 266$ there is a $[266,9,\dE=195]_4$ code, while the unique decoding radius for these parameters is $\left\lfloor \frac{1}{2} \cdot \frac{q}{q-1}r \right\rfloor =148$.
	\end{itemize}
	It is notable that the used code is only required to have good code parameters, but we do not need an efficient decoding algorithm.
	
	\begin{remark}
		Apart from the decoding attacks mentioned in~\ref{sec:attacks}, the chosen Goppa code has to resist structural attacks, i.e., attacks that recover the secret key from the public generator matrix, such as the attack on certain quadratic wild Goppa codes~\cite{AttackWild2017} or a potential attack resulting from the distinguisher on high-rate Goppa codes~\cite{AttackWild2013FUOPT}. Similar to the original McEliece system, the public generator matrix of the interleaved system is a generator matrix of a Goppa code, hence the same considerations apply. 
	\end{remark}

	\subsection{Key Size of repaired interleaved McEliece}

	Table~\ref{tab:keysize} compares the $(n,k,t)$ Wild McEliece and our proposed repaired $\ell$-interleaved $(n,k,\tpub)$ McEliece for typical SL (i.e., 128, 256 bits) in terms of the key size. 
	
	For each parameter set we compute the size of the public key in systematic form as $k(n-k)$ bits. Note that we assume appropriate padding and randomizing (so-called \emph{CCA2-conversion}) that protects against \emph{semantic attacks}, i.e., attacks where the plaintext is obtained from the systematic part.
	\begin{table}[h!]
		\setlength{\tabcolsep}{2pt}
		\def\arraystretch{1.1}
		\begin{center}
			\begin{tabular}{c|c|c|c|c|c|c|c|c|c}
				\hline
				\hline
				SL & \multirow{2}{*}{$q$} & \multirow{2}{*}{$m$} & \multirow{2}{*}{Method} & \multirow{2}{*}{$r$} & \multirow{2}{*}{$n$} & \multirow{2}{*}{$k$} & $t$  & \multirow{2}{*}{\footnotesize Rate} & Key size\\
				$[$bits$]$ & & &  &  &  &  & $(\ell,\tpub,\dE)$ & & [Bytes]\\
				\hline
				\multirow{7}{*}{$128$}& {$2$} & {$12$} & U.~D. & {$70$} & $2800$ & $1960$ & $70$ & $0.70$ & $205\ 800$ \\
                \cline{2-10}
				& \multirow{2}{*}{$3$} & \multirow{2}{*}{$8$} & U.~D. & \multirow{2}{*}{$100$} & $2420$ & $1620$ & $75$ & $0.67$ & $256\ 763$ \\
				& & & \textbf{Int.} & & $2130$ & $1330$ & $(7,131,84)$ & $0.62$ & $\mathbf{210\ 800}$ \\
				\cline{2-10}
				& \multirow{2}{*}{$4$} & \multirow{2}{*}{$6$} & U.~D. & \multirow{2}{*}{$90$} & $2150$ & $1610$ & $60$ & $0.75$ & $217\ 350$ \\
				& & & \textbf{Int.} & & $1580$ & $1040$ & $(7,105,76)$ & $0.66$ & $\mathbf{140\ 400}$ \\
				\cline{2-10}
				& \multirow{2}{*}{$5$} & \multirow{2}{*}{$5$} & U.~D. & \multirow{2}{*}{$100$} & $1800$ & $1380$ & $62$ & $0.74$ & $200\ 266$ \\
				& & & \textbf{Int.} & & $1290$ & $790$ & $(7,109,84)$ & $0.61$ & $\mathbf{114\ 646}$ \\
				\hline
				\multirow{7}{*}{$256$}& {$2$} & {$13$} & U.~D. & {$120$} & $6740$ & $5180$ & $120$ & $0.77$ & $1\ 010\ 100$ \\
				\cline{2-10}
				& \multirow{2}{*}{$3$} & \multirow{2}{*}{$8$} & U.~D. & \multirow{2}{*}{$180$} & $5100$ & $3660$ & $135$ & $0.72$ & $1\ 044\ 173$ \\
				& & & \textbf{Int. }& & $4300$ & $2860$ & $(7,236,156)$ & $0.67$ & $\mathbf{815\ 939}$ \\
				\cline{2-10}
				& \multirow{2}{*}{$4$} & \multirow{2}{*}{$7$} & U.~D. & \multirow{2}{*}{$240$} & $4880$ & $3200$ & $160$ & $0.66$ & $1\ 344\ 000$ \\
				& & & \textbf{Int. }&  & $3760$ & $2080$ & $(7,280,208)$ & $0.55$ & $\mathbf{873\ 600}$ \\
				\cline{2-10}
				& \multirow{2}{*}{$5$}& \multirow{2}{*}{$6$} & U.~D. & \multirow{2}{*}{$200$} & $4690$ & $3490$ & $125$ & $0.74$ & $1\ 215\ 530$ \\
				& & & \textbf{Int.} & & $3200$ & $2000$ & $(7,218,171)$ & $0.63$ & $\mathbf{696\ 578}$ \\
				\hline
				\hline
			\end{tabular}
		\end{center}
		\tiny{\ \ \ \ U.D.~$=$ Unique Decoding~\cite{McE78,WildBern2011}. \quad Int.~$=$ Interleaved Decoding (this paper).}
		\caption{Key size of repaired interleaved McEliece and wild McEliece for 128 and 256-bits security level (determined by ISD algorithm over $\F_q$~\cite{interlando2018generalization}). }\label{tab:keysize}
	\end{table}

	\vspace{-0.8cm}
	\section*{Acknowledgment}	
	
	The authors would like to thank Jean-Pierre Tillich for his comments that helped improve the quality of this work.

	\bibliographystyle{IEEEtran}
	\bibliography{main}
	
\end{document}